\documentclass[11pt]{article}
\pdfoutput=1
\usepackage[utf8]{inputenc}

\usepackage{jhep pub}

\usepackage[T1]{fontenc}
\usepackage{amssymb,amsmath,amsthm,amsfonts}
\usepackage{enumerate,mathrsfs}
\usepackage{xcolor}
\usepackage{cancel}
\usepackage[squaren]{SIunits}
\usepackage{hyperref}
\usepackage{mathtools}
\usepackage{prettyref}
\usepackage{blkarray}
\usepackage{adjustbox}
\usepackage{array,booktabs,tabularx,fancybox}
\usepackage{comment}
\usepackage[normalem]{ulem}
\usepackage{afterpage}
\usepackage{tikz}
\usetikzlibrary{decorations.pathmorphing}
\usetikzlibrary{decorations.markings}
\usetikzlibrary{calc}
\usetikzlibrary{fit, shapes, tikzmark}
\usepackage{float}

\usepackage{lipsum}
\newtheorem{theorem}{Theorem}

\newtheorem{conjecture}{Conjecture}
\newtheorem{lemma}{Lemma}

\definecolor{cyan1}{HTML}{37cdaa}
\definecolor{blue1}{HTML}{5d7ac4}
\definecolor{red1}{HTML}{d0482a}
\definecolor{purple1}{HTML}{845ea8}
\definecolor{orange1}{HTML}{e07229}
\definecolor{yellow1}{HTML}{edcb52}
\definecolor{green}{HTML}{244819}
\definecolor{red}{HTML}{921818}
\definecolor{purple}{HTML}{53047A}
\definecolor{yellow}{HTML}{f4e097}
\definecolor{zoomFuchsia}{HTML}{CB29BA}
\definecolor{zoomYellow1}{HTML}{FFE5B0}
\definecolor{zoomYellow2}{HTML}{FFD47B}
\definecolor{zoomYellow3}{HTML}{D08D00}
\definecolor{zoomPurple1}{HTML}{DBC9FF}
\definecolor{zoomPurple2}{HTML}{C2A4FF}
\definecolor{zoomPurple3}{HTML}{5300F9}
\definecolor{zoomGreen1}{HTML}{C2EBCE}
\definecolor{zoomGreen2}{HTML}{99DDAD}
\definecolor{zoomGreen3}{HTML}{339951}
\definecolor{zoomBlue1}{HTML}{ABD1FF}
\definecolor{zoomBlue2}{HTML}{73B2FF}
\definecolor{zoomBlue3}{HTML}{005AC8}
\definecolor{zoomRed1}{HTML}{FBC2C6}
\definecolor{zoomRed2}{HTML}{F899A0}
\definecolor{zoomRed3}{HTML}{D70F1E}
\definecolor{zoomOrange1}{HTML}{FFCEA7}
\definecolor{zoomOrange2}{HTML}{FFC291}
\definecolor{zoomOrange3}{HTML}{C15500}
\definecolor{gr}{gray}{0.7}
\definecolor{gr1}{gray}{0.7}

\newcommand{\brk}[1]{(#1)}

\DeclarePairedDelimiterX\braket[2]{\langle}{\rangle}{#1 \>\delimsize\vert\> #2}

\definecolor{green1}{HTML}{244819}

\newcommand{\cD}{\mathcal{D}}

\newcommand{\cI}{\mathcal{I}}

\title{D-ideal of generic mass banana integrals in dimensional regularization}
\author[b,c]{Wojciech Flieger,}
\date{}

\affiliation[a]{Dipartimento di Fisica e Astronomia, Universit\`a degli Studi di Padova, Via Marzolo 8, I-35131 Padova, Italy}
\affiliation[b]{INFN, Sezione di Padova, Via Marzolo 8, I-35131 Padova, Italy}

\emailAdd{flieger@pd.infn.it}

\abstract{We present a set of $\ell + 3$ simple differential operators and prove that they annihilate an $\ell$-loop generic mass banana integral in dimensional regularization. We study the singular locus of the ideal generated by these operators and show that it is contained in the set of Landau singularities of the first and second type. Through the Macaulay matrix method, we calculate the corresponding holonomic rank up to $\ell = 8$, obtaining $2^{\ell +1}-1$, which agrees with the number of master integrals for the generic mass banana integrals. Based on these findings, we conjecture that the proposed operators generate the annihilating ideal for the generic mass banana integrals in dimensional regularization.}   

\begin{document}

\maketitle

\section{Introduction}
Feynman integrals are ubiquitous in high-energy physics computations, ranging from particle physics to gravity and cosmology. Their efficient computation is crucial for connecting theory with precise experiments and to validate the theoretical models. Feynman integrals are complicated multi-valued functions belonging to a larger class of so-called twisted period integrals providing deep and intriguing relation to mathematics.

As functions of external parameters, i.e. Lorentz invariants and masses of the particles involved, they satisfy a system of differential equations which was first noticed in studies of the analytic structure of the scattering amplitudes \cite{Regge:1965,Giffon:1969se,Barucchi:1973}. In the modern context of dimensionally regulated Feynman integrals, differential equations have been first studied in~\cite{Kotikov:1990kg, Kotikov:1991pm,Remiddi:1997ny,Gehrmann:1999as,Henn:2013pwa}. To derive the differential system one relies on the so-called integration-by-parts identities (IBP) \cite{Tkachov:1981wb,Chetyrkin:1981qh,Laporta:2000dsw} which provide relations between Feynman integrals with the same set of propagators raised to different powers, the so-called family. It has been proved that all integrals of a given family can be expressed as a linear combination of finitely many so-called Master Integrals (MI) \cite{Smirnov:2010hn}, thus introducing the vector space structure for Feynman integrals. Recently, this fact combined with the insight from the theory of twisted period integrals and cohomology theory has been used to bypass the IBP identities and obtain coefficients of the linear relation, the so-called twisted intersection numbers, directly from the inner product of the vector space structure \cite{Mastrolia:2018uzb,Frellesvig:2019kgj,Frellesvig:2019uqt}.   

The system of differential equations can be recast into one scalar differential equation involving higher-order derivatives, known as the Picard-Fuchs (PF) equation. This can be viewed as a differential operator annihilating a given Feynman integral suggesting an approach focused on differential operators. Indeed such differential operators annihilating a given function can be studied by algebraic methods in the theory of $\mathcal{D}$-modules. Recently, the first applications of these algebraic methods have been applied to Feynman integrals \cite{Chestnov:2022alh,Chestnov:2023kww,Henn:2023tbo}. The crucial step in this approach is the construction of annihilators for Feynman integrals. Griffiths-Dwork pole reduction of rational differential forms \cite{Griffiths_1969,Dwork1962,Dwork1964} provides a framework for such construction. 
However, the straightforward application of the Griffiths-Dwork reduction to Feynman integrals introduces the boundary term, which has to be integrated out, resulting in an inhomogeneous differential equation.
This approach has already been applied to obtain inhomogeneous Picard-Fuchs operators for Feynman integrals \cite{Muller-Stach:2012tgj,Lairez:2022zkj,delaCruz:2024xit}.  The idea to bypass the presence of the boundary term was suggested in \cite{Golubeva:1973,Golubeva:1978} and applied to construct second-order operators for simple examples in finite dimensions. Recently built on this idea an algorithm for constructing annihilating differential operators for general dimensionally regulated Feynman integrals has been proposed in \cite{Chestnov:2025whi} opening a door for systematic application of $\mathcal{D}$-module techniques for Feynman integrals. Let us stress that the information contained in $\mathcal{D}$-modules is equivalent to that of the Pfaffian system, which is the usual way of representing differential equations for Feynman integrals, or Picard-Fuchs operator, and algorithms to switch between these different representations exist \cite{saito2000gröbner,Chestnov:2022alh, Chestnov:2025whi,gorlach2025}. Thus, $\mathcal{D}$-module techniques provide an alternative way for constructing a system of differential equations for Feynman integrals by bypassing the generation of IBP identities.  

Particularly interesting are holonomic $\mathcal{D}$-modules which have finite-dimensional space of solutions and Feynman integrals fall into this category \cite{Kashiwara:1977nf}.
An important example of holonomic systems coming from mathematics is a generalized hypergeometric system introduced by Gelfand-Kapranov-Zelevinsky (GKZ) for generalized Euler integrals \cite{Gelfand1989,Gelfand:1990}. Due to the underlying geometry and symmetries of the GKZ system, it is possible to read off annihilating differential operators directly from the integral representation and the holonomic rank of the system corresponds to the dimension of the cohomology group. Currently, for Feynman integrals treated as a restriction of generalized Euler integrals \cite{Chestnov:2022alh,Chestnov:2023kww} such a simple construction is not known.

Recently a lot of studies have been done for a family of $\ell$-loop two-point Feynman integrals, known as the banana integrals. These are the simplest examples of Feynman integrals for which the underlying geometry is more complicated and hence cannot be expressed just in terms of polylogarithmic functions. Already at two loops one can observe appearance of elliptic curves and integrals \cite{Sabry:1962rge,Broadhurst:1993mw,Laporta:2004rb,Bloch:2013tra,Adams:2015ydq,Adams:2017ejb,Broedel:2017siw,Adams:2018yfj,Honemann:2018mrb,Bogner:2017vim,Weinzierl:2020fyx,Remiddi:2016gno,Berends:1993ee,Caffo:1998du,Muller-Stach:2011qkg,Adams:2013nia,Remiddi:2013joa,Campert:2020yur,Bogner:2019lfa,Bloch:2016izu,Adams:2015gva,Adams:2014vja,Giroux:2022wav,Cacciatori:2023tzp} and for higher loops more complicated functions related to Calabi-Yau manifolds appear \cite{Broedel:2021zij,Pogel:2022yat,Pogel:2022ken,Duhr:2024bzt,Duhr:2025tdf,Duhr:2025kkq,Pogel:2025bca}. This requires us to extend our understanding of the underlying geometries and functions to compute efficiently these integrals. 
Recently interesting results on geometric, differential and analytic structure for equal and generic mass banana integrals in two dimensions \cite{Vanhove:2014wqa,Klemm:2019dbm,Bonisch:2020qmm} and dimensional regularization \cite{Bonisch:2021yfw} have been obtained.

In this work, we provide for the first time a set of $\ell+3$ simple differential operators that annihilate the $\ell$-loop generic mass banana integral in dimensional regularization and study its properties.  

The paper is organized as follows. In Section \ref{sec:math}, we provide the necessary mathematical background for our study including Griffiths-Dwork pole reduction, Macaulay matrix method and basics of the algebraic $\mathcal{D}$-module theory. In Section \ref{sec:banana}, we define $\ell$-loop banana integrals, introduce the set of $\ell+3$ differential operators and prove that they constitute the annihilating ideal for the integral. Further, we show that the singular locus of this ideal is contained in the set of Landau singularities of the first and second type and via the Macaulay matrix method we compute the holonomic rank up to $\ell = 8$ loops. Finally, we present conclusions and outlook in Section \ref{sec:conclusion}.

\section{Mathematical framework}\label{sec:math}
This section provides the necessary mathematical background for the studies of differential properties of Feynman integrals. None of the material is new and we do not present any proofs of the statements in this section for which we refer to original papers on the pole reduction \cite{Griffiths_1969,Dwork1962,Dwork1964}, Macaulay matrix method \cite{Chestnov:2022alh} and introductory textbooks on the theory of $\mathcal{D}$-modules \cite{coutinho1995primer,saito2000gröbner,hibi2014gröbner}.  

\subsection{Griffiths-Dwork pole reduction}
We will make use of the Griffiths-Dwork pole reduction for rational functions in the projective space \cite{Griffiths_1969,Dwork1962,Dwork1964}. In particular, we will be interested in the twisted version of the Griffiths-Dwork reduction which accounts for dimensionally regulated Feynman integrals \cite{delaCruz:2024xit}.
The starting point of the reduction is the following differential form
\begin{align}
\gamma = \frac{1}{\mathcal{F}^{k-1}} \sum_{i<j} (-1)^{i+j} \left( \alpha_{i}\lambda_{j} - \alpha_{j}\lambda_{i} \right) \frac{{\rm U}^{\kappa}}{{\rm F}^{\eta}}  d\alpha_{1} \wedge \hdots \widehat{d\alpha_{i}} \wedge \hdots \wedge \widehat{d\alpha_{j}} \hdots \wedge d\alpha_{n} \,,   
\end{align}
where $\rm U$ and $\rm F$ are homogeneous polynomials.
By taking the total derivative w.r.t. integration variables $\alpha_{i}$ we get
\begin{align}\label{eq:twisted_griffiths}
d \gamma = (\eta + k-1)\frac{\sum_{i=1}^{n} \lambda_{i}\frac{\partial {\rm F}}{\partial \alpha_{i}} }{{\rm F}^{k}}\Omega -  \frac{\sum_{i=1}^{n} \left( \frac{\partial \lambda_{i}}{\partial \alpha_{i}} + \kappa \lambda_{i} \frac{\partial \log({\rm U})}{\partial \alpha_{i}} \right) }{{\rm F}^{k-1}}\Omega \,, 
\end{align}
where $\Omega = \frac{{\rm U}^{\kappa}}{{\rm F}^{\eta}}  \sum _{i=1}^{n} (-1)^{i-1}\alpha_{i} \, d\alpha_{1} \wedge \cdots \wedge \widehat{d\alpha_{i}} \wedge \cdots \wedge d\alpha_{n}$, which is the sought after reduction formula. One can see that one is able to reduce the pole of ${\rm F}$ modulo the exact form $d \gamma$. The necessary condition for this reduction is that the numerator $\mathcal{N}$ of a given rational function belongs to the Jacobian ideal of $\rm{F}$, i.e. $\mathcal{N} \in \langle \frac{\partial {\rm F}}{\partial \alpha_{1}}, \hdots ,  \frac{\partial {\rm F}}{\partial \alpha_{n}} \rangle$  .
Recently, an algorithm has been established for constructing annihilating differential operators for Feynman integrals through the Griffiths-Dwork reduction by imposing additional boundary condition $\lambda_{i} \vert_{\alpha_{i}} =0$ which implies the vanishing of $d \gamma$ after integration \cite{Chestnov:2025whi}.

\subsection{D-modules}

Let us consider the $m$-th Weyl algebra 
\begin{align}
\mathcal{D}_{m} = \mathbb{C} [x_{1}, \ldots , x_{m}] \langle \partial_{1}, \ldots , \partial_{m}  \rangle \,,     
\end{align}
i.e. a ring of differential operators with polynomial coefficients, i.e. a noncommutative ring generated by $x_{1}, \hdots , x_{m}$ and $\partial_{1}, \hdots, \partial_{n}$ subjected to the following relations
\begin{align}
[x_{i}, x_{j}] = [\partial_{i}, \partial_{j}] = 0  \text{ and } [\partial_{i}, x_{j}] = \delta_{ij} \,. 
\label{eq:poly_weyl_commutation}
\end{align}
One can also define the $m$-th Weyl algebra with rational coefficients
\begin{align}
\mathcal{R}_{m} = \mathbb{C} (x_{1}, \ldots , x_{m}) \langle \partial_{1}, \ldots , \partial_{m}  \rangle \,,     
\end{align}    
which is a ring of differential operators with rational coefficients subjected to the commutation relations
\begin{align}
[\partial_{i}, \partial_{j}] = 0  \text{ and } [\partial_{i}, c(x)] = \frac{\partial c(x)}{\partial x_{i}} \text{ for } c \in \mathbb{C} (x_{1}, \ldots , x_{m}) \,. 
\label{eq:rat_weyl_commutation}
\end{align}
Working with $\mathcal{R}_{m}$ might be more suitable for practical calculations in particular for determining the holonomic rank. 

Any element $L$ of $\mathcal{D}_{m}$ can be written in a normally ordered form as
\begin{align}
  L = \sum_{(p,q) \in supp(L)} 
  c_{p,q} \, x^{p} \, \partial^{q}\,,     
  \label{eq:normally_ordered_form}
\end{align}
where we introduced multi-index notation $x^{p} = x_{1}^{p_{1}}\cdots x_{m}^{p_{m}}$ and similarly for $\partial^{q}$.
A non-empty subset $\mathcal{I}$ of $\mathcal{D}$ is called a left ideal of $\mathcal{D}$ if the following two conditions hold: \\
1) If $f,g \in \mathcal{I}$ then $f+g \in \mathcal{I}$,\\
2) If $f \in \mathcal{I}$ and $d \in \mathcal{D}$ then $df \in \mathcal{I}$. \\
% We are interested in the left modules over $\mathcal{R}$ and in particular in the modules of the form $\mathcal{R}/\mathcal{I}$ where $\mathcal{I}$ is a left ideal in $\mathcal{R}$.
A left module $M$ over $\mathcal{D}$ ($\mathcal{D}$-module) is an additive group for which an action of $\mathcal{D}$ is defined by the following map
\begin{align}
\mathcal{D} \times M \ni (d,m) \mapsto dm=\tilde{m} \in M    
\end{align}
which satisfies the following conditions \\
1) $1m = m$, $m\in M$, \\ 
2) $d_{1}(d_{2}m) = (d_{1}d_{2})m$, $d_{1},d_{2}\in \mathcal{D}$, $m \in M$, \\
3) $(d_{1}+d_{2})m = d_{1}m + d_{2}m$, $d_{1},d_{2} \in \mathcal{D}$, $m \in M$, \\
4) $d(m_{1}+m_{2})=dm_{1} + dm_{2}$, $d\in \mathcal{D}$, $m_{1},m_{2}\in M$. \\
Let $\mathcal{I}$ be a left ideal of $\mathcal{D}$, then $\mathcal{D}/\mathcal{I}$ can be regarded as a left $\mathcal{D}$-module. 

Let $u,v \in \mathbb{R}^{m}$, be two real $m$-dimensional vectors. If $u_{i}+v_{i} \geq 0$ for any $i$, we call the pair $(u,v)$ the weight vector in $\mathcal{R}$. Let us define the order of the monomials w.r.t. the weight vector as
\begin{align}
{\rm ord}_{(u,v)}(x^{p} \partial^{q}) = u \cdot p + v \cdot q \,,     
\end{align}
where $u \cdot p$ stands for the inner product. For $L \in \mathcal{R}$ we define
\begin{align}
{\rm ord}_{(u,v)}(L) = max_{(p,q)} {\rm ord}_{(u,v)} (x^{p} \partial^{q}) \,.   
\end{align}
If $u_{i} + v_{i} > 0 $ for any $i$ we define the initial term of $L$ to be
\begin{align}
{\rm in}_{(u,v)}(L) = \sum_{(p,q), u\cdot p + v \cdot q = m} c_{p,q} x^{p} \xi^{q} \,,     
\end{align}
where $x_{i}$ and $\xi_{j}$ are now commutative variables and ${\rm in}_{(u,v)}(L)$ is an element of the ring of polynomials $\mathbb{C}[x,\xi]$.
The principal symbol of the operator $L \in \mathcal{D}$ is the sum of its highest differential operators
\begin{align}
{\rm in}_{(0,\mathbf{1})}(L) = \sum_{\substack{(p,q) \in {\rm supp}(L), \\ 
\vert q \vert={\rm max}}} c_{p,q} \, x^{p} \,  \xi^{q}\,.   
\end{align}
where $\mathbf{1}= (1,\hdots,1)$. The characteristic variety of the ideal $\cI$ in $\mathcal{D}$ is the zero set of the initial ideal ${\rm in}_{(0,\mathbf{1})}(\mathcal{I})$,
\begin{align}
{\rm char}(\mathcal{I}) = \lbrace (x,\xi)\in \mathbb{C}^{2m} \vert p(x,\xi) = 0, p \in {\rm in}_{(0,\mathbf{1})}(\mathcal{I})    \rbrace \,.    
\end{align}
% The characteristic variety of $\mathcal{I}$ is equal to the characteristic variety of the left $\mathcal{D}$-module $\mathcal{R}/\mathcal{R} \mathcal{I}$.
The projection of the characteristic variety ${\rm char}(\mathcal{I})\backslash \lbrace \xi_{1}=\ldots=\xi_{m}=0 \rbrace$ to the coordinate $x$-space is called the singular locus ${\rm sing}(\mathcal{I})$ of the ideal $\mathcal{I}$.
Of particular interest are holonomic $\cD$-modules for which the dimension of the characteristic variety is equal to $m$. The holonomic rank $r$ of the ideal $\cI$ is defined as the dimension of the following vector space over the field $\mathbb{C}(x)$
\begin{align}
    r(\cI) = \text{dim}_{\mathbb{C}\brk{x}}\mathbb{C}(x)[\xi]/\left( \mathbb{C}(x)[\xi] \cdot 
    {\rm char}(\mathcal{I}) \right)  = \text{dim}_{\mathbb{C}\brk{x}} \left( \mathcal{R} / \mathcal{R} \mathcal{I} \right)\,.
    \label{eq:holonomic_rank}
\end{align}

\subsection{Macaulay matrix method}
In the Macaulay matrix method, one builds the linear system of equations where the unknowns are monomials in derivatives $\partial^{q}$. The system is constructed by acting with derivatives on the ideal $\mathcal{I} \in \mathcal{R}$ 
\begin{align}
M = (\partial^{I} \mathcal{I}) \,, \quad \vert I \vert = 0,\hdots, k \,,  
\end{align}
expanded in the normally ordered form, where $l$ is a chosen maximal order of the derivative, i.e. $\partial^{I}= \partial_{x_{1}}^{i_{1}}\cdots \partial_{x_{m}}^{i_{m}}$ with $i_{1} + \hdots + i_{m} = \vert I \vert$. The matrix of the coefficients $M$ of this system is called the Macaulay matrix. By solving this system one identifies the set of independent set monomials in derivatives known as the standard monomials (Std). The number of standard monomials corresponds to the holonomic rank $r$ of the ideal $\mathcal{I}$ 
\begin{align}
\vert {\rm Std} \vert = r(\mathcal{I}) \,.    
\end{align}
This method allows us to find the independent set of generators of the $\mathcal{D}$-module which is an analogue of the IBP reduction for differential operators. If the maximal order $k$ is chosen high enough one can construct the Picard-Fuchs operator w.r.t. a chosen variable $x_{i}$.

\section{Banana integrals}\label{sec:banana}

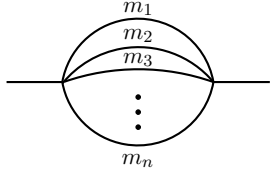
\begin{figure}[t]
\centering
\begin{tikzpicture}
\coordinate (v1) at (0,0);
\coordinate (v22) at (2,0);
\coordinate (pout) at (2.75,0);
\coordinate (pin) at (-0.75,0);
\coordinate (pp1) at (1,-0.2);
\coordinate (pp2) at (1,-0.4);
\coordinate (pp3) at (1,-0.6);
\coordinate (n2) at (1,1);
\draw[thick ] (pin) -- (v1);
\draw[thick ] (pout) -- (v22);
\draw[thick] (2,0) arc (10:170:1.02) node at (1,0.98) [scale=0.8] {$m_{1}$};
\draw[thick] (2,0) arc (40:141:1.30) node at (1,0.62) [scale=0.8] {$m_{2}$};
\draw[thick] (2,0) arc (70.6:109.2:3) node at (1,0.3) [scale=0.8] {$m_{3}$};
%\draw[ultra thick] (2,0) arc (-40:-140:1.29) ;
\draw[thick] (2,0) arc (-10:-170:1.02) node[midway,below,scale=0.8] {$m_{n}$};
%
%\draw[fill=black] (v1) circle (.06cm);
%\draw[fill=black] (v22) circle (.06cm);
%
\draw[fill=black] (pp1) circle (.03cm);
\draw[fill=black] (pp2) circle (.03cm);
\draw[fill=black] (pp3) circle (.03cm);
\end{tikzpicture}
\caption{The graph corresponding to the generic mass $\ell$-loop banana integral.}
\label{fig:banana}
\end{figure}

We consider a family of $\ell$-loop two-point Feynman integrals with generic masses in a parametric representation corresponding to the diagram in Fig.~\ref{fig:banana}
\begin{align}\label{eq:banana_integral}
I(x) = \int_{\Gamma} \frac{{\rm U}^{n-\frac{\ell+1}{2}d}}{{\rm F}^{n - \frac{\ell d}{2}}} \mu\equiv \int_{\Gamma} \Omega \,, 
\end{align}
where $x=(x_{0}, \hdots , x_{n})$, $
\mu \equiv \sum _{i=1}^{n} (-1)^{i-1}\alpha_{i} \, d\alpha_{1} \wedge \cdots \wedge \widehat{d\alpha_{i}} \wedge \cdots \wedge d\alpha_{n}$, $\Gamma= \lbrace (\alpha_{1}, \hdots , \alpha_{n}) \vert
\sum_{i=1}^{n}\alpha_{i} =1 \,, \alpha_{i} \geq 0 \rbrace$ and $n=\ell+1$ is the number of internal lines.
The Symanzik polynomials are given by
\begin{align}
&{\rm U} =  \left( \prod_{k=1}^{n} \alpha_{k} \right) \left( \sum_{j=1}^{n} \frac{1}{\alpha_{i}} \right)\\
&{\rm F} = \prod_{i=1}^{n}\alpha_{i} x_{0} - {\rm U} \sum_{i=1}^{n} \alpha_{i} x_{i} \,,
\end{align}
which are homogeneous polynomials respectively of degree $\ell$ and $\ell + 1$ w.r.t. integration variables. The polynomial ${\rm F}$ is also homogeneous of degree $1$ w.r.t. external parameters $x$ making the integral \eqref{eq:banana_integral} a homogeneous function of degree $\frac{\ell d}{2} -n$.

\subsection{Annihilators of banana integrals}
The algorithm proposed in \cite{Chestnov:2025whi} allows for the construction of differential operators annihilating Feynman integrals and to use them to generate the annihilating $\mathcal{D}$-ideal. However, the operators produced by the algorithm might not be presented in the most compact way, possibly obstructing general patterns and properties. By studying and refining the one and two-loop operators coming from Griffiths-Dwork construction, we were able to extract a set of simple operators valid for all loop generic mass banana integrals.  
\begin{theorem}\label{prop:operators}
The following differential operators annihilate the $\ell$-loop banana integral
\begin{align}
1) \,\,\,\, &\mathcal{A}_{E} = \sum_{i=0}^{n} x_{i}\partial_{i} + \left(n - \frac{n-1}{2}d \right) \,, \\
2) \,\,\,\, &\mathcal{A}_{i} = - x_{0}\partial_{0}^{2} + x_{i} \partial_{i}^{2} -\frac{d}{2} \partial_{0} + \left(2-\frac{d}{2} \right)\partial_{i}    \,, \quad  i=1, \ldots , n \,, \\
3) \,\,\,\, &\mathcal{A}_{s} = \left( \prod_{i=0}^{n} \partial_{i} \right) \left( \sum_{i=0}^{n} \frac{1}{\partial_{i}} \right)  \,.  
\end{align}
\end{theorem}
\begin{proof}
\hfill \break
1) The Euler operator $\mathcal{A}_{E}$ annihilates integrand
\begin{align}
\begin{split}
\mathcal{A}_{E} \Omega =& \frac{-\left(n - \frac{n-1}{2}d\right)\left( x_{0}\prod_{i=1}^{n}\alpha_{i} - {\rm U}\sum_{i=1}^{n} \alpha_{i} x_{i} \right)}{{\rm F}}\Omega + \left(n - \frac{n-1}{2}d \right) \Omega  \\
=&\frac{-\left(n - \frac{n-1}{2}d \right){\rm F}}{{\rm F}}\Omega + \left(n - \frac{n-1}{2}d \right) \Omega = 0 \,. 
\end{split}
\end{align}
2) We will show the annihilation property via the Griffiths-Dwork reduction. Without loss of generality let us choose $i=n$ and let us act with the second-order part of the operator on the integral
\begin{align}
\left( - x_{0}\partial_{0}^{2} + x_{n} \partial_{n}^{2} \right) I(x) = \int_{\Gamma} \frac{\mathcal{N}_{2}}{{\rm F}^{2}} \Omega \,, 
\end{align}
where
\begin{align}
\mathcal{N}_{2} = \left(  n - \frac{n-1}{2}d \right) \left(  n - \frac{n-1}{2}d +1 \right)\left( -x_{0} \prod_{k=1}^{n}\alpha_{k}^{2} + \alpha_{n}^{2} x_{n} U^{2}  \right) \,.   
\end{align}
Let us show that $\mathcal{N}_{2} \in \langle \frac{\partial {\rm F}}{\partial \alpha_{1}} \,, \hdots \,, \frac{\partial {\rm F}}{\partial \alpha_{n}}  \rangle$ with the following coefficients
\begin{align}
&\lambda_{j} = \frac{1}{n}\left(  n - \frac{n-1}{2}d \right) \left(  n - \frac{n-1}{2}d +1 \right) \alpha_{n} \alpha_{j} \left( U - \frac{1}{\alpha_{n}}\prod_{k=1}^{n} \alpha_{k} \right) \quad j = 1, \hdots , n-1 \,, \\
&\lambda_{n} = -\frac{1}{n}\left(  n - \frac{n-1}{2}d \right) \left(  n - \frac{n-1}{2}d +1 \right) \alpha_{n}^{2}  \left( (n-1) U + \frac{1}{\alpha_{n}}\prod_{k=1}^{n} \alpha_{k} \right) \,.
\end{align}
Then,
\begin{align}
\begin{split}
\sum_{i=1}^{n} \lambda_{i} \frac{\partial {\rm F}}{\partial \alpha_{i}} =& \frac{1}{n}\left(  n - \frac{n-1}{2}d \right) \left(  n - \frac{n-1}{2}d +1 \right)\times \\  
&\bigg[
\alpha_{n} U \sum_{i=1}^{n}\alpha_{i} \frac{\partial {\rm F}}{\partial \alpha_{i}} - \prod_{k=1}^{n} \alpha_{k} \sum_{i=1}^{n}\alpha_{i} \frac{\partial {\rm F}}{\partial \alpha_{i}} - n \alpha_{n}^{2} U \frac{\partial {\rm F}}{\partial \alpha_{n}}
\bigg]  \\
=&\left(  n - \frac{n-1}{2}d \right) \left(  n - \frac{n-1}{2}d +1 \right) \left[  
-\prod_{k=1}^{n}\alpha_{k}^{2} x_{0} +\alpha_{n}^{2} x_{n} U^{2}
\right] \,,
\end{split}
\end{align}
where we used $\sum_{i=1}^{n}\alpha_{i} \frac{\partial {\rm F}}{\partial \alpha_{i}} = n {\rm F}$ and $\alpha_{n} \frac{\partial {\rm F}}{\partial \alpha_{n}} = {\rm F} +\frac{1}{\alpha_{n}}\prod_{k=1}^{n}\alpha_{k}\sum_{l=1}^{n}\alpha_{l}x_{l} -\alpha_{n}x_{n} {\rm U}$. Thus
\begin{align}
\mathcal{N}_{2} =  \sum_{i=1}^{n} \lambda_{i} \frac{\partial {\rm F}}{\partial \alpha_{i}} \,.   
\end{align}
Since each $\lambda_{i}$, $i=1, \hdots , n$ is proportional to $\alpha_{i}$ we have for the $\gamma$ in the Griffiths-Dwork reduction
\begin{align}
\int_{\Gamma} d \gamma = \int_{\partial \Gamma} \gamma = 0 \,.    
\end{align}
Thus, we have the following relation
\begin{align}
\frac{\mathcal{N}_{2}}{{\rm F}^{2}} \Omega = \frac{1}{\left( n -\frac{n-1}{2}d +1 \right)} \frac{\sum_{i=1}^{n}\frac{\partial \lambda_{i}}{\partial \alpha_{i}} + \sum_{i=1}^{n} \lambda_{i} \frac{\partial}{\partial \alpha_{i}} \log\left( {\rm U}^{n - \frac{n}{2}d}\right)}{{\rm F}}\Omega \,.   
\end{align}
Let us calculate the numerator of the r.h.s.
\begin{align}
\begin{split}
&\sum_{i=1}^{n} \left( \frac{\partial \lambda_{i}}{\partial \alpha_{i}} + \lambda_{i} \frac{\partial}{\partial \alpha_{i}} \log\left( {\rm U}^{n - \frac{n}{2}d}\right) \right) = \frac{1}{n}\left( n -\frac{n-1}{2}d \right) \left( n -\frac{n-1}{2}d +1 \right) \times \\
&\bigg[  
-2n\prod_{k=1}^{n}\alpha_{k} -n\alpha_{n}^{2}\frac{\partial {\rm U}}{\partial \alpha_{n}} + \left( n -\frac{n}{2}d \right)\alpha_{n}\sum_{i=1}^{n} \alpha_{i} \frac{\partial {\rm U}}{\partial \alpha_{i}} \\
&- \left( n -\frac{n}{2}d \right)\frac{1}{{\rm U}}\prod_{k=1}^{n} \alpha_{k} \sum_{i=1}^{n} \alpha_{i} \frac{\partial {\rm U}}{\partial \alpha_{i}} - \left( n -\frac{n}{2}d \right) n \alpha_{n}^{2} \frac{\partial {\rm U}}{\partial \alpha_{n}}
\bigg]  \\=
&-\left( n -\frac{n-1}{2}d \right) \left( n -\frac{n-1}{2}d +1 \right) \bigg[
\frac{d}{2}\prod_{k=1}^{n} \alpha_{k} + \left( 2 - \frac{d}{2} \right) \alpha_{n} {\rm U}
\bigg] \,,
\end{split}
\end{align}
where we used $\sum_{i=1}^{n}\alpha_{i} \frac{\partial {\rm U}}{\partial \alpha_{i}} = (n-1) {\rm U}$ and $\alpha_{n}\frac{\partial {\rm U}}{\partial \alpha_{n}} = {\rm U} - \frac{1}{\alpha_{n}}\prod_{k=1}^{n}\alpha_{k}$. 
Let us apply the first order part of the operator to the integral
\begin{align}
\left( \frac{d}{2} \partial_{0} - \left( 2 - \frac{d}{2} \right) \partial_{n} \right) I(x) = \int_{\Gamma} \frac{\mathcal{N}_{1}}{{\rm F}} \Omega \,, 
\end{align}
where
\begin{align}
\mathcal{N}_{1} = - \left( n - \frac{n-1}{2}d \right) \left[
\frac{d}{2} \prod_{k=1}^{n} \alpha_{k} + \left( 2-\frac{d}{2} \right)\alpha_{n} {\rm U}
\right] \,.
\end{align}
Thus
\begin{align}
\frac{1}{\left( n -\frac{n-1}{2}d +1 \right)}\sum_{i=1}^{n} \left( \frac{\partial \lambda_{i}}{\partial \alpha_{i}} + \lambda_{i} \frac{\partial}{\partial \alpha_{i}} \log\left( {\rm U}^{n - \frac{n}{2}d}\right) \right) = \mathcal{N}_{1} \,,   
\end{align}
and
\begin{align}
\int_{\Gamma} \left( \frac{\mathcal{N}_{2}}{{\rm F}^{2}} - \frac{\mathcal{N}_{1}}{{\rm F}} \right) \Omega  = 0 \,.   
\end{align}

\noindent 3) The $\mathcal{A}_{s}$ operator annihilates the integrand
\begin{align}
\begin{split}
\mathcal{A}_{s} \Omega =& (-1)^{n} \prod_{i=0}^{n-1}(\eta + i) \left( (-1)^{n-1}\prod_{i=1}^{n} \alpha_{i} {\rm U}^{n-1} \prod_{i=1}^{n}\alpha_{i}\sum_{i=1}^{n}\frac{1}{\alpha_{i}} + (-1)^{n} {\rm U}^{n} \prod_{i=1}^{n} \alpha_{i} \right) \frac{\Omega}{{\rm F}^{n}}  \\
=&(-1)^{2n-2}\prod_{i=0}^{n-1}(\eta +i) \left( \prod_{i=1}^{n}\alpha_{i}{\rm U}^{n} - \prod_{i=1}^{n}\alpha_{i}{\rm U}^{n}  \right) \frac{\Omega}{{\rm F}^{n}} = 0 \,.
\end{split}
\end{align}
\end{proof}
% This pattern has been first observed for $\ell=1,2$ from the modified Griffiths-Dwork reduction \cite{Chestnov:2025whi}.

Moreover, by direct computation we get,
\begin{lemma}
Operators $\mathcal{A}_{E}$, $\mathcal{A}_{i}$ and $\mathcal{A}_{s}$ satisfy the following commutation relations
\begin{align}
\begin{split}
&[\mathcal{A}_{E},\mathcal{A}_{i}] = -\mathcal{A}_{i} \,, \quad \quad [\mathcal{A}_{E},\mathcal{A}_{s}] = -n\mathcal{A}_{s} \,, \\
&[\mathcal{A}_{i},\mathcal{A}_{j}] = 0 \,, \quad \quad \quad \,\,\,\, [\mathcal{A}_{i},\mathcal{A}_{s}] = (\partial_{0} - \partial_{i})\mathcal{A}_{s} \,.
\end{split}
\end{align}
\end{lemma}

\subsection{Singular locus}
\begin{theorem}
The singular locus of the ideal $\mathcal{I}_{B}$ generated by the operators $\mathcal{A}_{E}$, $\mathcal{A}_{s}$ and $\mathcal{A}_{i}$, $i=1,\hdots n$, is contained in
\begin{align}\label{eq:sing_loc}
&{\rm sing}(\mathcal{I}_{B}) \subseteq \prod_{i=0}^{n} x_{i} \prod_{\lbrace \delta_{\pm} \rbrace} \left(x_{0}-\left(\sum_{j=1}^{n} \delta_{\pm} \sqrt{x_{j}} \right)^{2} \right) \,, \text{ where } \delta_{\pm} = \pm 1 \,,
\end{align}
and the second product runs over all projectively nonequivalent ways to assign $\delta_{\pm}$ in the sum.
\end{theorem}
\begin{proof}
To get the singular locus we want to eliminate $\xi$ variables from the characteristic variety. Let us notice that there are $n+2$ equations in $n+1$ variables $\xi$, thus one of them is redundant.
From the symbols of the operators $\mathcal{A}_{i}$ we get
\begin{align}
x_{0} \xi_{0}^{2} = x_{1} \xi_{1}^{2} = \hdots = x_{n} \xi_{n}^{2} \,.   
\end{align}
From this we get for example
\begin{align}
\sqrt{x_{0}} = \pm \sqrt{x_{n}}\frac{\xi_{n}}{\xi_{0}}, \hdots , \sqrt{x_{n-1}} = \pm \sqrt{x_{n}}\frac{\xi_{n}}{\xi_{n-1}} \,.     
\end{align}
By adding these equation we get
\begin{align}
\sqrt{x_{0}} + \sum_{j=1}^{n-1}\sqrt{x_{j}} = \pm \sqrt{x_{n}} \xi_{n} \sum_{k=0}^{n-1} \frac{1}{\xi_{k}} = \mp \sqrt{x_{n}} \,,
\end{align}
where we took the common denominator and used the symbol of the operator $\mathcal{A}_{s}$. Hence,
\begin{align}
x_{0} =  \left( \sum_{i=1}^{n-1} \sqrt{x_{i}} \pm \sqrt{x_{n}}   \right)^{2} \,.    
\end{align}
By repeating this procedure for all $x_{i}$, $i=0,\hdots, n$, we obtain all solutions for the case $\xi_{j} \neq 0$, $j=0,\hdots n$, which gives
\begin{align}
\prod_{\lbrace \delta_{\pm} \rbrace} \left(x_{0}-\left(\sum_{j=1}^{n} \delta_{\pm} \sqrt{x_{j}} \right)^{2} \right) = 0 \,.    
\end{align}
Now, by considering some of $\xi_{j}$ equal to zero we get
\begin{align}
\prod_{i=0}^{n} x_{i} = 0 \,.   
\end{align}
Thus the singular locus of the ideal $\mathcal{I}_{B}$ is contained in
\begin{align}
\text{sing}(\mathcal{I}_{B}) \subseteq  \prod_{i=0}^{n} x_{i} \prod_{\lbrace \delta_{\pm} \rbrace} \left(x_{0}-\left(\sum_{j=1}^{n} \delta_{\pm} \sqrt{x_{j}} \right)^{2} \right) \,.     
\end{align}
\end{proof} 
We have not been able to prove the opposite inclusion, however, the explicit computation of the singular locus of $\mathcal{I}_{B}$ for low $\ell$ suggests that equality holds in the above theorem.
The result agrees with the formula for singularities of banana integrals provided without proof in \cite{Ponzano:1969tk} and also with the leading singularity part coming from Landau discriminant \cite{Mizera:2021icv}. Recently, in \cite{Matsubara-Heo:2025lrq} the formula for the hypergeometric discriminant of banana integrals was proposed, whose projection to the kinematic space describes the above singular locus.

\subsection{Holonomic Rank}
We compute via the Macaulay matrix method the holonomic rank of the $\mathcal{D}$-ideal $\mathcal{I}_{B}$ up to $\ell=8$. The results are collected in the table below:
\begin{table}[H]
    \centering
    \begin{tabular}{|c|c|c|c|c|c|c|c|c|}
    \hline
        $\ell$ & 1 & 2 & 3 & 4 & 5 & 6 & 7 & 8 \\ \hline
        $r(\mathcal{I}_{B})$ & 3 & 7 & 15 & 31 & 63 & 127 & 255 & 511 \\
    \hline    
    \end{tabular}
    %\caption{}
    %\label{tab:my_label}
\end{table}
\noindent This agrees with the known formula $2^{\ell+1}-1$ for number of MI's for banana integrals \cite{Bitoun:2017nre}.

The results of this section motivate us to state the following,
\begin{conjecture}
The holonomic $\mathcal{D}$-ideal of rank $2^{\ell+1}-1$ for the $\ell$-loop banana integral is generated by the following $\ell+3$ operators
\begin{align}
&\mathcal{A}_{E} = \sum_{i=0}^{n}x_{i} \partial_{i} + \left(n -\frac{n-1}{2}d \right) \,, \\
&\mathcal{A}_{i} = - x_{0}\partial_{0}^{2} + x_{i} \partial_{i}^{2} -\frac{d}{2} \partial_{0} + \left(2-\frac{d}{2} \right)\partial_{i}    \,, \quad  i=1, \ldots , n \,, \\
&\mathcal{A}_{s} = \left( \prod_{i=0}^{n} \partial_{i} \right) \left( \sum_{i=0}^{n} \frac{1}{\partial_{i}} \right)  \,.
\end{align}
\end{conjecture}

\section{Conclusions and outlook}\label{sec:conclusion}
Banana Feynman integrals play a central role in the understanding of the mathematical properties of Feynman integrals. We proposed a set of $\ell +3$ simple differential operators and proved that they annihilate the $\ell$-loop banana integral. We also showed that the singular locus of the ideal generated by these operators is contained in the set of Landau singularities of the banana integrals. Moreover, through the Macaulay matrix method, we computed the holonomic rank of this ideal for $\ell = 8$, obtaining the agreement with the well-known number $2^{\ell +1}-1$ of master integrals for the banana integrals. This evidence motivates us to conjecture that these operators generate the annihilating $\mathcal{D}$-ideal for the banana integrals. 
This opens a completely new avenue in studies of Feynman integrals 
from a simple set of differential operators. As such, this $\mathcal{D}$-ideal and related $\mathcal{D}$-module techniques provide an alternative way for deriving a system of differential equations by bypassing the generation of IBP identities.

The next natural step in the studies of $\mathcal{D}$-modules for Feynman integrals will be the proof of the proposed conjecture. Since the operators are simple and reflect a certain pattern one should ask if there is any underlying symmetry dictating the form of these operators similarly to the GKZ system. It would be interesting to see if a similar set of differential operators exists for other types of Feynman integrals for example ladder diagrams or non-planar cases.

\section*{Acknowledgments}
The author thanks Vsevolod Chestnov, Pierpaolo Mastrolia, Saiei-Jaeyeong Matsubara-Heo, Nobuki Takayama and William J. Torres~Bobadilla for valuable discussions, comments on the manuscript and collaboration on a joint project which inspired this work, and Janusz Gluza for comments on the manuscript. 
The author would like to thank organizers and participants of the Domoschool 2025: "Amplitools: mathematical methods for Particle Physics, Gravitation and Cosmology", where related results have been presented.
The author would like to acknowledge the support of the INFN research initiative {\it Amplitudes} and of the Polish National Science Centre (NCN) under grant 2023/50/A/ST2/00224.

\appendix

\bibliographystyle{JHEP}
\bibliography{refs.bib}

\end{document}